\documentclass[10pt,onecolumn,twoside]{IEEEtran}

\IEEEoverridecommandlockouts                              

\usepackage{amsmath}
\usepackage{amssymb}
\usepackage{amsfonts}

\usepackage{amsthm}
\usepackage{comment}
\usepackage{cite}
\usepackage[ruled,vlined]{algorithm2e}

\usepackage{setspace}
\usepackage{hyperref}

\usepackage[color=green!40]{todonotes}

\usepackage{wrapfig}

\theoremstyle{plain}
\newtheorem{assumption}{Assumption}
\newtheorem{remark}{Remark}

\newtheorem{lemma}{Lemma}
\newtheorem{corollary}{Corollary}
\theoremstyle{plain}
\newtheorem{theorem}{Theorem}
\theoremstyle{definition}
\newtheorem{definition}{Definition}
\theoremstyle{definition}

\theoremstyle{definition}
\newtheorem{problem}{Problem}

\title{\LARGE \bf
Socioeconomic Impact of Emerging Mobility Markets and Implementation Strategies
}

\author{Ioannis Vasileios Chremos, \textit{Student Member, IEEE}, and Andreas A. Malikopoulos, \textit{Senior Member, IEEE}%
\thanks{This research was supported by the Sociotechnical Systems Center (SSC) at the University of Delaware.}%
\thanks{The authors are with the Department of Mechanical Engineering, University of Delaware, Newark, DE 19716 USA (emails: \tt\small{ichremos@udel.edu}; \tt\small{andreas@udel.edu}.)}%
}

\begin{document}

\maketitle

\begin{abstract}

Emerging mobility systems such as connected and automated vehicles (CAVs) provide the most intriguing opportunity for more accessible, safe, and efficient transportation. CAVs are expected to significantly improve safety by eliminating the human factor and ensure transportation efficiency by allowing users to monitor transportation network conditions and make better operating decisions. However, CAVs could alter the users' tendency-to-travel, leading to a higher traffic demand than expected, thus causing rebound effects (e.g., increased vehicle-miles-traveled). In this chapter, we focus on tackling social factors that could drive an emerging mobility system to unsustainable congestion levels. We propose a mobility market that models the economic in-nature interactions of the travelers in a smart city network with roads and public transit infrastructure. Using techniques from mechanism design, we introduce appropriate monetary incentives (e.g., tolls, fares, fees) and show how a mobility system consisting of selfish travelers that seek to travel either with a CAV or use public transit can be socially efficient. Furthermore, the proposed mobility market ensures that travelers always report their true travel preferences and always benefit from participating in the market; lastly, we also show that the market generates enough revenue to potentially cover its operating costs.

\end{abstract}

\section{Introduction}
\label{sec:introduction}

Nowadays, it is nearly impossible to commute in a major urban area without the frustration of congestion and traffic jams. Moreover, congestion is one of the leading factors behind road accidents and altercations, negatively impacting the economic success of cities and the quality of life of their citizens. For these reasons, congestion has been broadly recognized as one of the major challenges to address for next-generation cities. One incoming highly transformative innovation that promises to address congestion though is autonomous driving, and in particular, connected and automated vehicles (CAVs). Recent advancements in emerging mobility systems with CAVs are highly expected to eliminate congestion and increase mobility efficiency in terms of energy consumption and travel time \cite{zhao2019_ITSMagazine}. In addition, CAVs are expected to have vast technological, commercial, and regulatory dimensions \cite{marletto2019}.

There has been a significant amount of work on the technological impact of CAVs, mostly focusing on congestion, emissions, energy consumption, and safety \cite{taiebat2018,zmud2017}. It is apparent that CAVs will transform today's urban transportation system and revolutionize mobility \cite{barnes2017}. However, one of the most novel and defining characteristics of an emerging mobility system is its socioeconomic complexity. Mobility is an indispensable prerequisite for social, cultural, and economic development as well as social participation. Thanks to the unprecedented improvements in mobility, we expect a significant alteration in human behavior and, most importantly, on tendency-to-travel. This may lead to unintended consequences, i.e., rebound effects, in the sense of additional energy use and greenhouse gas emissions, as well as leading to decreases in the density of urban areas and negatively impacting congestion. In addition, future mobility systems will enable human-vehicle interactions between people of any age and abilities, thus allowing enhanced and universal accessibility. One key reason why connectivity (e.g., Internet of Things) and automation in mobility may lead to rebound effects is because of the high levels of comfort and convenience - factors that urge drivers, passengers, and travelers to change their commute and travel tendencies, and thus use their vehicles quite more frequently and more unexpectedly. As urban social life has been greatly associated with the technological impact of the car, this compels us to reassess the relationship between automobility and social life \cite{sheller2000,bissell2020}. To add to our argument, evident from similar technological revolutions, for example, the impact of elevators on building design and social class hierarchies \cite{bernard2014}, human social perspective and view can have a tremendous effect on how technological innovations are utilized and implemented. For all these reasons, it is vital to study the impact of CAVs in a sociotechnical context focusing on the social implications and attempt to provide optimal solutions for the efficient CAV-utilization in society.


There is a solid body of research now available for optimizing the efficiency of emerging mobility systems with CAVs. Over the last decade, several research efforts reported in the literature \cite{zardini2020,malikopoulos2018_Automatica,malikopoulos2021_Automatica,zhang2019,zhang2018} have aimed at addressing questions regarding the CAVs' impact on transportation efficiency. For example, can we consider the problem of optimizing fuel economy and emissions by coordinating a mobility system consisting of CAVs? What would be the appropriate conceptual approaches for modeling and optimizing emerging mobility systems? Recent technological developments can answer the above questions, indicating that CAVs will most likely help us eliminate congestion, significantly decrease fuel consumption, and minimize road accidents. Analytical frameworks have been proposed to quantify and evaluate the impacts of CAVs from the technological perspective \cite{jackeline2016_ITSC,jackeline2016_IWCTS}. Furthermore, coordination of CAVs at different traffic scenarios (e.g., intersections, vehicle-following) have been extensively evaluated in the literature \cite{jackeline2017_ITS,xiao2021,beaver2020_VSD,mahbub2020_Automatica,jang2019_ICCPS}. Moreover, the impact of CAVs has been identified as one that will enable traffic administrators to monitor transportation network conditions efficiently and effectively, thus improving the operating decisions that are required daily \cite{sarkar2016,zhao2019_ITSMagazine}. However, they are challenges. The cyber-physical nature of emerging mobility systems is associated with significant control challenges and gives rise to a new level of complexity in modeling and control \cite{ferrara2018}. These challenges tend to focus on the technological dimension, and what is mostly missing is a complementary study to the broader social implications of CAVs. For example, the impact of selfish social behavior in routing networks of regular and autonomous vehicles has been studied \cite{mehr2019_TCNS,lazar2020,biyik2021}, as well as ``how people learn'' in mobility systems with behavioral dynamics \cite{krichene2015,krichene2016,cabannes2018_ITSC,lam2016}. However, it seems that the problem of how CAVs will affect human tendency-to-travel and decision-making has not been adequately approached yet. Understanding this ``social'' aspect of CAVs is critical in our effort to design efficient mobility systems.

One of the standard approaches to alleviate congestion in a transportation system has been the management of demand size due to the shortage of space availability and scarce economic resources in the form of congestion pricing (alternatively called ``tolling mechanisms'' \cite{pigou2013,vickrey1969}). Such an approach focuses primarily on intelligent and scalable traffic routing, in which the objective is to guide and coordinate users in path-choice decision-making. For example, one computes the shortest path from a source to a destination regardless of the changing traffic conditions \cite{schmitt2006}. Interestingly, by adopting a game-theoretic approach, advanced systems have been proposed to assign users concrete routes or minimize travel time and studying the Nash equilibria under different tolling mechanisms \cite{chen1998,joksimovic2004,wada2011,salazar2019,chremos2020_CDC,chremos2021_ECC}. This motivates us to ask: ``How can we design an emerging mobility system that ensures that all travelers reach their destination safely, efficiently, and in a timely manner?'' This question is quite important as it is widely accepted that CAVs will revolutionize the way people travel. We aim to provide a first-attempt answer to this question in this chapter and argue that a sociotechnical approach focusing on the social dimension of a mobility problem can help us design the next-generation mobility systems. To achieve this, we consider a mobility system with decentralized information (alternatively called ``asymmetric information'') and multiple selfish and intelligent decision-makers (e.g., travelers), who, in turn, may misreport their true travel preferences for better individual benefits. Hence, based on their background and unique behavioral tendencies, travelers make decisions that generally do not lead to system-wide optimal performance. We tackle this \emph{discrepancy between individual and collective interests} \cite{chremos2020_ITSC} by reverse-engineering the mobility system from its optimal solution (e.g., efficiency, congestion-free) to what should each traveler do via the implementation of monetary incentives. This method in economics is known as ``mechanism design,'' in which by treating systems as economic institutions, we can control and coordinate the selfish agents' ``economic activity'' (e.g., which mode of transportation to use).

The theory of mechanism design was developed as an objective-first approach to efficiently align the individuals' and system's interests in problems of asymmetric information, where the individual agents have private preferences \cite{mas_colell1995,nisan2007}. It can be viewed as the art of designing the rules of a game to achieve a specific desired outcome. A well-established and broadly-used mechanism that has been successful in widely different applications (e.g., auctions, public projects, and cost-minimization problems) is the Vickrey-Clarke-Groves (VCG) mechanism \cite{vickrey1961,clarke1971,groves1973}. The VCG mechanism ensures the existence and implementation of a dominant strategy equilibrium, which is an efficient solution and allows selfish agents to make a decision (alternatively choose a strategy) that is best no matter what other agents may decide. Agents are also incentivized to truthfully report their private preferences and have no reason (e.g., chance of receiving negative utility) not to participate in the mechanism. However, the VCG mechanism is known to be an extravagant mechanism, i.e., it can generate big surpluses. Overall, mechanism design has broad applications spanning surprisingly many different fields, including microeconomics, social choice theory, and control engineering. Applications in engineering include communication networks \cite{renou2012}, social networks \cite{dave2020_TAC}, transportation routing \cite{bian2020}, online advertising \cite{kakade2013}, smart grid \cite{samadi2012}, multi-agent systems \cite{shoham2008}, and in general resource allocation problems \cite{hurwicz2006}. We provide a formal overview of mechanism design in Section \ref{sec:theoretical-preliminaries}.

The application of mechanism design is not new in transportation and mobility problems \cite{iwanowski2003,teodorovic2008,vasirani2011,raphael2015,olarte2018}. For example, it has been used to provide solutions to individual route selection under different congestion traffic scenarios (e.g., first-mile ridesharing, selfish routing, tradable driving permits). In particular, auction-based mechanisms treat traffic congestion as an economic problem of supply and demand, focusing on travel time allocation or routing. So, on the one hand, auctions have been proposed to design pricing schemes with tolls in a network of roads leading to a spark of studies in auctioning techniques. On the other hand, this approach has important limitations: (i) the implementability of auction-based tolling on highways is not straightforward due to the dynamic and fast-changing nature of transportation systems; (ii) it is also uncertain how the public (e.g., drivers, passengers, travelers) will respond concerning toll roads in an auction setting. Therefore, understanding the travelers' interests (willingness-to-pay, value of time) and the impacts on different sociodemographic groups become imperative for a socially-efficient design of an emerging mobility system. For these reasons, it is essential to design an emerging mobility system whose focal point is the social aspect and societal impact of CAVs. In conjunction, it is the authors' belief that the emerging mobility systems - CAVs, shared mobility, electric vehicles - will be characterized by their socioeconomic complexity: (1) improved productivity and energy efficiency, (2) widespread accessibility, and (3) drastic urban redesign and evolved urban culture. This characteristic can naturally be modeled and analyzed using game theory/mechanism design and behavioral economics alongside control and optimization techniques. One of the main arguments in this chapter is that the social interactions of human travelers with CAVs, and other modes of transportation can be modeled as an economically-inspired mobility market, where monetary incentives (tolls) are used to induce the desired socially-efficient outcome.

Our aim is to develop a holistic and rigorous framework to capture the societal impact of connectivity and automation in emerging mobility systems and provide solutions that prevent any potential rebound effects (e.g., increased vehicle-miles-traveled, increased travel demand, empty trips). To achieve this aim, as a first attempt, we study an emerging mobility system consisting of a finite group of travelers who seek to travel in a ``smart city,'' where a central authority (alternatively called social planner) seeks to ensure the efficient distribution and operation of the different modes of transportation offered by the city. We call these different modes of transportation ``mobility services.'' A few examples of mobility services are CAVs, shared vehicles, and public transit (e.g., train, bus, light rail, subway). The travelers request to use at most one service to satisfy their mobility needs, i.e., to reach their destination, via a smartphone app easily accessible to all travelers. The social planner (e.g., a central computer) compiles all travelers' origin-destination requests and other information (e.g., preferred travel time, value of time, and maximum willingness-to-pay) in order to provide a travel recommendation to each traveler. The social planner's goal is to ensure that the aggregate travel recommendations are \emph{socially-efficient}. Informally, by socially-efficient, we mean that the endmost collective travel recommendation must achieve two objectives: (i) respect and satisfy the travelers' preferences regarding mobility, and (ii) ensure the alleviation of congestion in the system. Since our focus is to provide socially-efficient solutions, we consider a city that supports connected and automated mobility technologies on its roads and public transit infrastructure. Subsequently, the social planner is fully aware of the system's capabilities and network's capacity. In other words, the social planner is fully capable of computing the maximum capacity of each mobility service and the associated costs aimed at providing travel recommendations to all travelers.

Our objective in this chapter is to design a mobility market of an emerging mobility system and provide a socially-efficient solution consisting of well-designed and appropriate monetary incentives (e.g., tolls, fares, fees) for a social planner to guarantee the realization of the desired outcome, i.e., maximize the social welfare of all travelers. At the same time, our solution will ensure to provide such incentives to travelers so that the usage of any mobility service will not lead to congestion in the mobility system. In other words, we design a mobility market that efficiently assigns each traveler to the ``right'' mode of transportation.


Our contributions are the following: we design a socially-efficient mobility market that assigns mobility services to a finite group of travelers by taking into consideration their travel preferences. We achieve that by implementing a special case of the VCG mechanism after modifying it accordingly for a mobility problem. We show that the proposed mobility market is incentive compatible and individually rational, two properties that ensure all selfish travelers are truthful in their communication with the social planner and voluntarily participate in the mobility market. We also show that the proposed market is economically sustainable, i.e., it generates revenue from each traveler and ensures that the operating costs of each mobility service are covered. It is through the appropriate design of monetary incentives that we successfully incentivize all travelers to truthfully report their travel preferences and voluntarily participate in the market. Thus, we are guaranteed a socially-efficient mobility solution. The proposed mobility market also provides an incentive to central authorities to implement it, since as we show, the market ensures that there are minimum acceptable payments to cover the operating costs of the mobility services.


The chapter is structured as follows. In Section \ref{sec:theoretical-preliminaries}, we review the main concepts of mechanism design and briefly discuss the VCG mechanism. In Section \ref{sec:mobility-market}, we present the mathematical formulation of the emerging mobility market, which forms the basis for the rest of the chapter. In Subsection \ref{subsec:problem-statement}, we present the imposed optimization problem. In Section \ref{sec:methodology}, we present the methodology used to design the monetary incentives for each traveler. In Section \ref{sec:properties}, we study the properties of the mobility market, and finally, in Section \ref{sec:conclusion}, we draw conclusions and offer a discussion for future research.

\section{Theoretical Preliminaries}
\label{sec:theoretical-preliminaries}

In this section, we provide the theoretical preliminary material related to this chapter's proposed modeling framework, and we formally introduce all important concepts needed to prove our principal results.

\subsection{An Introduction to Mechanism Design}
\label{subsection:intro_mech_design}

Most generic control systems can be viewed as a specification of how decisions (e.g., how to utilize a number of resources) are determined as a function of the information that is known by the agents in the system. What interests us in most cases is \emph{efficiency}, i.e., realizing the best possible allocation of resources with the best use of information to achieve an outcome where collectively agents are satisfied, and there is no overutilization of the system's resources \cite{maskin2002}. One key challenge in ensuring efficiency in a control system is the fact that different agents may have conflicting interests and act selfishly. In other words, systems that incorporate human decision-making, if remained uninfluenced, are not guaranteed to exhibit optimal performance. This is well-known to be the case in control theory, and economics \cite{myerson2008,brown2017}. There are various different theories and approaches that attempt to guarantee efficiency in such systems and can provide solutions of varying degrees of success. One such theory is mechanism design, in which we are concerned with how to implement system-wide optimal solutions to problems involving multiple selfish agents, each with private information about their preferences \cite{nisan2001,diamantaras2009}. Within the context of mobility, agents are the travelers, and their private information can be either tolerance to traffic delays, value of time, preferred travel time, or any disposition to a specific mode of transportation. Our goal in mechanism design is to design appropriate incentives in order to align the interests of agents with the interests of the system \cite{hurwicz2006}. For example, in mobility, given that each traveler/driver/passenger ``competes'' with everyone else to reach their destination first, we want to ensure that given this inherent conflict of interest, we can still guarantee uncongested roads, no traffic accidents, and no travel time delays. Mechanism design can help us design the rules of systems where information is decentralized (different agents know different things), and agents do not necessarily have an immediate incentive to cooperate \cite{borgers2015}. In particular, mechanism design helps us design rules that align all agents' decision-making by providing the right incentives to achieve a well-defined objective for the system (e.g., aggregate optimal performance, system-level efficiency). Thus, mechanism design entails solving an optimization problem with sometimes unverifiable and always incomplete information structure \cite{maskin2008}. We call such a problem \emph{an incentive design and preference elicitation problem}.

We start by supposing that there is a system consisted of a finite group of agents, each competing with each other for a limited, fixed allocation of resources. Each agent evaluates different allocations based on some private information that is known only to them. We consider a \emph{social planner}, playing the role of a centralized entity, whose task is to align the selfish and conflicting interests of the agents with the overall system's objective (e.g., an efficient allocation of resources or the maximization of social welfare). As it can be seen in Fig. \ref{fig:mechanism_design_visualization}, there are four components: (1) There is a group of decision-makers, (2) who make a decision based on their personal information, and (3) that decision is reported as a message to the social planner who is tasked to design the rules of which (4) it can be determined what each agent gets. What follows next is a mathematically formal presentation of the social planner's task.

\begin{figure}
    \centering
    \includegraphics[scale = 0.3]{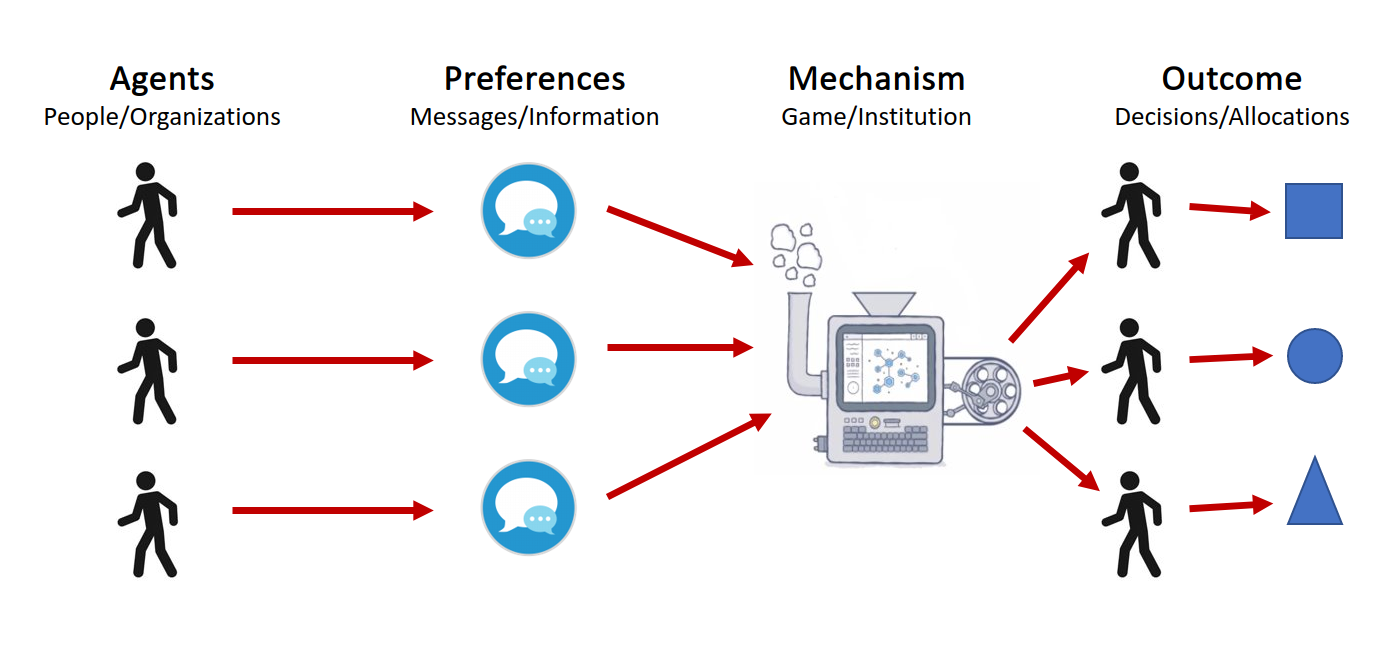}
    \caption{A visualization of how an arbitrary control system (agents, preferences, allocations) can be viewed under a mechanism design framework. Agents hold private information, of which they send reports to the social planner who is responsible for designing a mechanism. How efficient the mechanism is can depend on whether the agents' messages are truthful or not.}
    \label{fig:mechanism_design_visualization}
\end{figure}

Consider a set of selfish agents $\mathcal{I}$, $|\mathcal{I}| = n \in \mathbb{N}$ with preferences over different outcomes in a set $\mathcal{O}$. Each agent Each agent $i \in \mathcal{I}$ is assumed to possess private information, denoted by $\theta_i \in \Theta_i$. Since an agent $i$'s $\theta_i$ can characterize and influence their decision-making in a significant way, we call $\theta_i$ the \emph{type} of agent $i$. We write $(\theta_i)_{i \in \mathcal{I}} = \theta \in \Theta = \prod_{i \in \mathcal{I}} \Theta_i$ to represent the type profile of all agents. Next, an agent $i$'s preferences over different outcomes can be represented by a utility function $u_i : \mathcal{O} \times \Theta_i \to \mathbb{R}$. Although the exact form of $u_i$ can vary depending on the application of the problem \cite{tavafoghi2016,tang2011,bitar2016,kazumura2020}, what is common in the literature \cite{shoham2008,borgers2015,nisan2007} is a quasilinear function of the form
    \begin{equation}\label{eqn:defn_utility}
        u_i(o, \theta_i) = v_i(o, \theta_i) - p_i,
    \end{equation}
where $v_i : \mathcal{O} \times \Theta_i \to \mathbb{R}_{\geq 0}$ represents an arbitrary valuation function, and $p_i \mapsto \mathbb{R}$ is a monotonically increasing function. If outcome $o \in \mathcal{O}$ represents an allocation of a resource, then $p_i$ can be thought of as a transfer of agent $i$'s wealth or a cost imposed to agent $i$ for that particular allocation $o$. Intuitively, a quasilinear function defined as in \eqref{eqn:defn_utility} ensures that the marginal value of $v_i$ does not depend on how large $p_i$ becomes, and vice-versa. Furthermore, \eqref{eqn:defn_utility} assumes $u_i$ is linear with respect to $p_i$. We can now naturally define the \emph{social welfare} as the collective summation of all agents' valuations, i.e.,
    \begin{equation}\label{eqn:social_welfare}
        w(o, \theta) = \sum_{i \in \mathcal{I}} v_i(o, \theta_i).
    \end{equation}
If our system objective is to maximize $w$, then immediately we observe that there is an important obstacle, i.e., any agent $i$ may misreport their type $\theta_i$ in the hopes to increase their own utility. So, the question is now: How can we incentivize agents to truthfully report their type? The answer is through the appropriate design of $p_i$. Next, we outline the building blocks that can help us design $p_i$. Formally, we can define a \emph{mechanism} as the tuple $\langle f, p \rangle$ composed of a \emph{social choice function} (SCF) $f : \Theta \to \mathcal{O}$ and a vector of \emph{payment functions} $p = (p_i)_{i \in \mathcal{I}}$, with $p_i : \Theta \to \mathbb{R}$. In words, a mechanism $\langle f, p \rangle$ defines the rules of which we can implement a system objective by mapping the agents' types to an outcome while using the payments to ensure the optimality or efficiency of that outcome (see Fig. \ref{fig:mechanism_design_framework} for an illustration of the mechanism design framework). We can now state the social planner's problem as follows
    \begin{gather}
        \max_{o \in \mathcal{O}} w(o, \theta) \label{eqn:obj_function} \\
        \text{subject to: } \hat{\theta}_i = \theta_i, \quad \forall i \in \mathcal{I}, \label{eqn:constraint_truth} \\
        \sum_{i \in \mathcal{I}} v_i(o, \theta_i) \geq \sum_{i \in \mathcal{I}} v_i(o ', \theta_i), \quad \forall o ' \in \mathcal{O}, \label{eqn:constraint_efficiency} \\
        \sum_{i \in \mathcal{I}} p_i(s(\theta)) \geq 0, \quad \forall \theta \in \Theta, \label{eqn:constraint_budget} \\
        v_i(f(s(\theta))) - p_i(s(\theta)) \geq 0, \quad \forall i \in \mathcal{I}, \forall \theta \in \Theta, \label{eqn:constraint_IR}
    \end{gather}
where $\hat{\theta}_i$ denotes the reported type of agent $i$, $s(\cdot)$ is the equilibrium strategy profile (e.g., Nash equilibrium). Constraints \eqref{eqn:constraint_truth} ensure the truthfulness in the agents' reported types, \eqref{eqn:constraint_efficiency} impose an efficiency condition, \eqref{eqn:constraint_budget} make certain that no external payments are required, and \eqref{eqn:constraint_IR} incentivize all agents to voluntarily participate in the mechanism. If we could know for certain the true types of all agents, then we would be able solve the optimization problem \eqref{eqn:obj_function} - \eqref{eqn:constraint_IR} using standard optimization techniques. However, as this is unreasonable to expect from selfish decision-makers, the social planner needs to elicit $\theta = (\theta_i)_{i \in \mathcal{I}}$ by designing the appropriate $p = (p_i)_{i \in \mathcal{I}}$. We discuss in the next subsection one such mechanism that elicits the private information of agents truthfully.

\begin{figure}
    \centering
    \includegraphics[scale = 0.58]{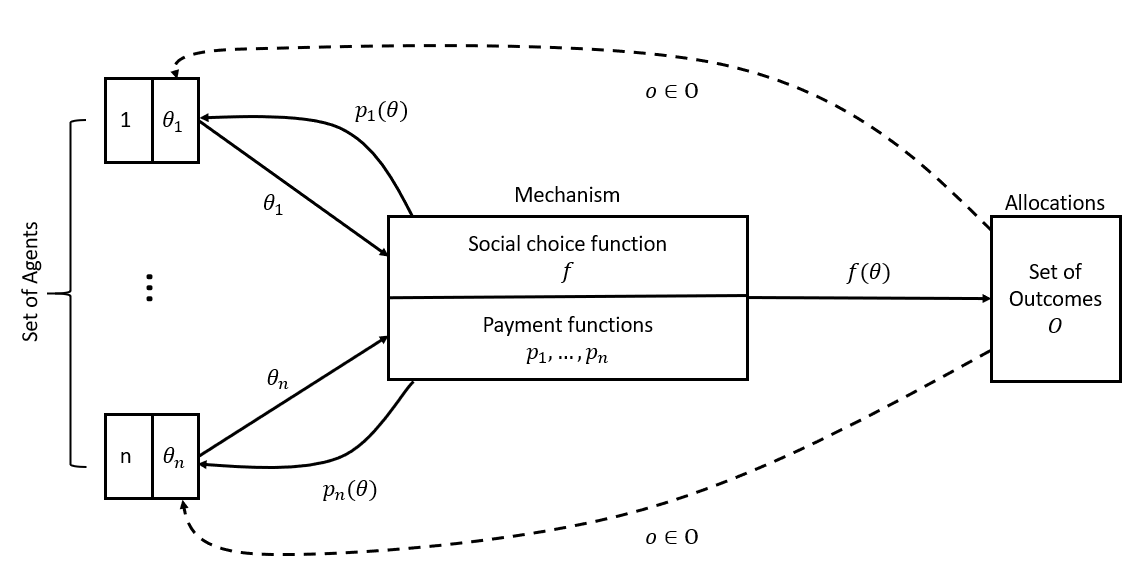}
    \caption{A theoretical representation of the mechanism design framework.}
    \label{fig:mechanism_design_framework}
\end{figure}

\subsection{The Vickrey-Clarke-Groves Mechanism}
\label{subsection:VCG_mechanism}

In the previous subsection, we reviewed the main concepts of mechanism design and formulated the incentive design and preference elicitation problem. In words, we asked ``How can we design the payments $p = (p_i)_{i \in \mathcal{I}}$ so that every agent makes the decision that agrees with what \emph{we} have chosen as the system's objective (e.g., efficiency)? To answer this question, in this subsection, we review the Vickrey-Clarke-Groves (VCG) mechanism \cite{vickrey1961,clarke1971,groves1973}, one of the most successful mechanisms as it incentivizes agents to be truthful and guarantees efficiency.

As we discussed earlier, a mechanism is a tuple $\langle f, p \rangle$. In a VCG mechanism, the SCF $f$ is defined as an allocation rule (who gets what) based on the optimization problem \eqref{eqn:obj_function} - \eqref{eqn:constraint_IR}, i.e.,
    \begin{equation}\label{eqn:scf}
        f(\hat{\theta}) = \arg \max_{o \in \mathcal{O}} W(o, \hat{\theta}_i).
    \end{equation}
where $\hat{\theta} = (\hat{\theta}_i)_{i \in \mathcal{I}}$. In words, assuming that the agents disclose their true information, \eqref{eqn:scf} provides to the social planner who attempts to maximize the social welfare a formal way to compute the allocations of each agent. At the same time, the VCG mechanism charges each agent for their allocation as follows
    \begin{equation}\label{eqn:payments}
        p_i(\hat{\theta}) = \sum_{j \neq i} v_j (f(\hat{\theta}_{- i})) - \sum_{j \neq i} v_j (f(\hat{\theta})),
    \end{equation}
where $\hat{\theta}_{- i}$ denotes the type profile of all agents except agent $i$. Note that the payments defined in \eqref{eqn:payments} do not depend on an agent $i$'s own declaration $\hat{\theta}_i$. Let us assume for a moment that all agents declare their types truthfully. Then, the first sum in \eqref{eqn:payments} computes the value of the social welfare with agent $i$ not participating in the mechanism. The second sum in \eqref{eqn:payments} computes the value of the social welfare of all other agents $j \neq i$ with agent $i$ participating in the mechanism. Thus, agent $i$ when they report $\hat{\theta}_i$ are made to pay the \emph{marginal effect} of their decision (in our case that is agent $i$'s reported type $\hat{\theta}_i$). In other words, this particular design of the payments in \eqref{eqn:payments} internalizes an agent $i$'s social externality, i.e., agent $i$'s impact on every other agents' welfare.

The VCG mechanism represented by the SCF $f$ defined by \eqref{eqn:scf} and the payment functions $p$ defined by \eqref{eqn:payments} satisfies the following properties:
    \begin{enumerate}
        \item For any agent, truth-telling is a strategy that dominates any other strategy that is available for that agent. We say then that truth-telling is a \emph{dominant strategy}. Note that such strategies are ``always optimal'' no matter what the other agents decide.
        \item The VCG mechanism successfully aligns the agents' individual interests with the system's objective. In our case, that objective was to maximize the social welfare of all agents. We call this property, \emph{economic efficiency}.
        \item For any agent, the VCG mechanism incentivizes them to voluntarily participate in the mechanism as no agent loses by participation (in terms of utility).
        \item The VCG mechanism ensures no positive transfers are made from the social planner to the agents. Thus, the mechanism does not incur a loss. We call this \emph{weakly budget balanced}.
    \end{enumerate}
The VCG mechanism essentially ensures the realization of a \emph{socially-efficient outcome}, i.e., satisfying properties 1 - 3, in a system of selfish agents, where each possesses private information. It is noteworthy to note how powerful the VCG mechanism is as it induces a dominant strategy equilibrium maximizing the social welfare while also making sure no agent is hurt by participating.

We conclude Section \ref{sec:theoretical-preliminaries} with the following remark: although the main motivation of mechanism design is the microeconomic study of institutions and relies heavily on game-theoretic techniques, it can prove a powerful theory providing a systematic methodology in the design of systems of asymmetric information, consisted of strategic decision-makers, and whose performance must attain a specified system objective. The rest of the chapter shall present how we can use this theory to design a socially-efficient mobility system consisting of travelers who compete with each other for the utilization of a limited number of mobility services.

\section{The Emerging Mobility Market}
\label{sec:mobility-market}

We consider an emerging mobility system consisting of a transportation city network managed by a social planner and a finite group of travelers who seek to travel in the network. Informally, this network represents the high-level mobility connections of multiple and different city neighborhoods. In other words, we move away from the concept of ``personally-owned'' modes of transportation and focus our modeling towards mobility provided as a service. This means that a social planner (e.g., a central computer) offers travelers a unified gateway of public and private transportation providers capable of providing mobility solutions to manage and realize their trip. For example, travelers can plan their journey via a smartphone app by specifying their preferences (e.g., cost, time, and convenience) and their desired destination. The social planner then is tasked to offer a travel recommendation to each traveler, i.e., which mode of transportation to take. In addition, we consider that multiple and different travel options can be offered to each traveler focusing on urban modes of transportation (e.g., CAVs, bus, train). We call these options ``mobility services'' or ``services'' for short. Within this framework, we propose a mobility market for a socially-efficient implementation of connectivity and automation in an emerging mobility system. The goal of the mobility market is twofold: (i) ensure that all travelers voluntarily participate and truthfully report their travel preferences, and (ii) be economically sustainable by generating revenue from each traveler and setting a minimum acceptable mobility payment for each traveler to potentially cover the operating costs.

\subsection{Mathematical Formulation of the Emerging Mobility Market}
\label{subsec:formulation}

The proposed mobility market is managed by a social planner who aims to allocate $m \in \mathbb{N}$ mobility services to $n \in \mathbb{N}$ travelers, where $n \geq m \neq 0$. We denote the set of travelers by $\mathcal{I}$, $|\mathcal{I}| = n$ and the set of mobility services by $\mathcal{J}$, $|\mathcal{J}| = m$. For example, each service $j \in \mathcal{J}$ can either represent a shared CAV, a train, or a bus. Both sets $\mathcal{I}$ and $\mathcal{J}$ are nonempty, disjoint, and finite. The set of all mobility services $\mathcal{J}$ can be partitioned to a finite number of disjoint subsets, each representing a specific ``type'' of a mobility service, i.e., $\mathcal{J} = \bigcup_{h = 1} ^ H \mathcal{J}_h$, where $H \in \mathbb{N}$ is the number of subsets of $\mathcal{J}$. For example, $\mathcal{J} = \mathcal{J}_{1} \cup \mathcal{J}_{2}$, where $|\mathcal{J}_1|$ represents the number of all available CAVs, and $|\mathcal{J}_2|$ represents the number of all available busses. Next, travelers seek to travel using these mobility services in a transportation network represented by an undirected multigraph $\mathcal{G} = (\mathcal{V}, \mathcal{E})$, where each node in $\mathcal{V}$ represents a different city area or neighborhood, and each link $e \in \mathcal{E}$ represents a sequence of city roads or a public transit connection. For our purposes, we think of $\mathcal{G} = (\mathcal{V}, \mathcal{E})$ as a representation of a smart city network with a road and public transit infrastructure. In $\mathcal{G}$, a traveler $i \in \mathcal{I}$ seeks to travel from their current location $o_i \in \mathcal{V}$ to their desired destination $d_i \in \mathcal{V}$. So, on one hand, each traveler $i \in \mathcal{I}$ is associated with a origin-destination pair $(o_i, d_i)$. On the other hand, each type of mobility services (e.g., one type is shared CAVs, another is trains) is associated with a unique link that connects any two nodes. At the same time, we do not limit the number of different mobility services that connect any origin $o_i$ to any destination $d_i$ of any traveler $i \in \mathcal{I}$. We suppose that any traveler $i \in \mathcal{I}$ has at least two travel options for their origin-destination pair $(o_i, d_i)$. Furthermore, each traveler $i \in \mathcal{I}$ can travel in $\mathcal{G}$ with any mobility service $j \in \mathcal{J}$ that satisfies their origin-destination pair $(o_i, d_i)$ and each service $j \in \mathcal{J}$ can be used by multiple travelers.

\begin{remark}\label{rmk:upper_level_graph}
    Network $\mathcal{G}$ represents the upper-level connections of different city neighborhoods. By connections, we mean either roads or public transit routes. Instead of modeling each node to represent travelers' exact location, we consider dividing a city into zones. By grouping travelers' exact locations into such zones, we can use network $\mathcal{G}$ to model the mobility connections between the different city zones.
\end{remark}

Next, we partition the set of travelers $\mathcal{I}$ into different smaller subsets characterized by a common origin-destination pair.

\begin{definition}\label{DEFN:traveler-partition}
    The set of travelers with the exact same origin-destination pair is $\mathcal{I}_k = \{i \in \mathcal{I} \; | \; (o_i, d_i) = (o_k, d_k)\}$, $k = 1, 2, \dots, K$, where $K \in \mathbb{N}$ is the number of subclasses over the complete set of travelers, i.e., $\mathcal{I} = \bigcup_{k = 1} ^ K \mathcal{I}_k$.
\end{definition}

The justification of Definition \ref{DEFN:traveler-partition} is that in an emerging mobility system, we can acquire verifiable location data of travelers either by using a global positioning system or estimating the average number of travelers using public transit \cite{coleri2004,tubaishat2009}.

Mathematically, the allocation of the finite number of mobility services to travelers can be described by a vector of binary variables.

\begin{definition}\label{defn:traveler_service_assign}
    The \emph{traveler-service assignment} is a vector $\mathbf{a} = (a_{i j})_{{i \in \mathcal{I}}, j \in \mathcal{J}}$, where $a_{i j}$ is a binary variable of the form:
        \begin{equation}\label{EQN:binary-variable}
            a_{i j} =
                \begin{cases}
                    1, \; & \text{if $i \in \mathcal{I}$ is assigned to $j \in \mathcal{J}$}, \\
                    0, \; & \text{otherwise}.
                \end{cases}
        \end{equation}
\end{definition}

Note that we have $(a_{i j})_{{i \in \mathcal{I}}, j \in \mathcal{J}} = (a_{1 1}, \dots, a_{i j}, \dots, a_{n m})$. By partitioning the set of travelers in $K \in \mathbb{N}$ subclasses, the traveler-service assignment of subclass $\mathcal{I}_k$ is given by $\mathbf{a}_k = (a_{i j})_{{i \in \mathcal{I}_k}, j \in \mathcal{J}}$.

Naturally, we need to impose a physical limit on the use of each mobility service $j \in \mathcal{J}$ in network $\mathcal{G}$ as well as a connection capacity of a mobility service for each link in the network. Note that each link in $\mathcal{G}$ represents a road or a public transit connection, which means that multiple mobility services of one type use that one link. For example, one link can be a bus lane with stops between two different city areas; another can be a train route between two stations.

\begin{definition}\label{DEFN:service-capacity}
    The \emph{usage capacity} of any mobility service $j \in \mathcal{J}$ is given by $\varepsilon_j \in \mathbb{N}$. The \emph{link capacity} in network $\mathcal{G}$ is given by $\gamma_e \in \mathbb{R}_{\geq 0}$.
\end{definition}

For example, $\varepsilon_j$ can represent the maximum number of travelers (or passengers) in a shared vehicle or the maximum number of travelers in a train vehicle (seated and standing). Similarly, $\gamma_e$ can represent a critical traffic density of mobility services, which means that any additional input of vehicles or trains can lead to a reduced traffic flow and eventually to traffic congestion. For example, we can use the GreenShields model to define explicitly the critical traffic density \cite{garber2014}.

As in any mobility problem that involves travelers, we need to consider the travelers' preferences (e.g., preferred travel time, value of time, willingness-to-pay for service). Hence, we formally define the notion of ``personal travel requirements'' by introducing three important parameters (our selection of those three parameters is justified by recent transportation studies \cite{polydoropoulou2020,ho2018}.)

\begin{definition}\label{defn:personal_travel_requir}
    For any traveler $i \in \mathcal{I}_k$, $k = 1, \dots, K$, let $\alpha_i \in (0, 1)$ be the \emph{value of time}, $\theta_i \in \mathbb{R}_{\geq 0}$ the \emph{preferred travel time}, and $\bar{v}_i \in \mathbb{R}_{\geq 0}$ the \emph{maximum willingness-to-pay}. Then, the \emph{personal travel requirements} of traveler $i$ is a tuple of the form $\pi_i = (\alpha_i, \theta_i, \bar{v}_i)$.
\end{definition}

We offer the intuition behind each parameter: traveler $i$'s value of time $\alpha_i$ transforms the traveler's time urgency in monetary units as it can model, for example, the acceptable amount of compensation for lost time. Similarly, a traveler $i$'s preferred travel time $\theta_i$ is a non-negative real value representing how fast traveler $i$ wishes to reach their destination. The last term in $\pi_i$ captures how much traveler $i$ appraises a direct and completely convenient mobility service. For example, $\bar{v}_i$ can measure the maximum willingness-to-pay of traveler $i$ traveling with the fastest and most convenient service (e.g., taking a taxicab with no co-travelers) to their destination.

For each traveler $i \in \mathcal{I}_k$, the tuple $\pi_i$ is considered private information, known only to traveler $i$. Hence, as the social planner does not know $(\pi_i)_{i \in \mathcal{I}}$, each traveler $i$ must report their $\pi_i$. This is one of the key challenges in the proposed mobility market: \emph{How can we incentivize the travelers to be truthful and elicit the private information needed to provide a socially-efficient solution to the emerging mobility market?} The answer to this question will be given in Section \ref{sec:methodology}.

Next, we introduce an ``inconvenience'' metric for any traveler $i \in \mathcal{I}_k$ using any mobility service $j \in \mathcal{J}$. Quantitatively, the inconvenience metric can represent the extra monetary value of travel disutility from any costs, travel delays, or violation of personal travel requirements caused by the use of a mobility service.

\begin{definition}\label{DEFN:inconvenience-metric}
    The \emph{mobility inconvenience metric} for traveler $i \in \mathcal{I}_k$, $k = 1, \dots, K$, assigned to service $j \in \mathcal{J}$ is a continuous, increasing, and convex function $\phi_i \left( \alpha_i, \theta_i, \tilde{\theta}_i(\mathbf{a}_k) \right) \mapsto \mathbb{R}_{\geq 0}$, where $\tilde{\theta}_i(\mathbf{a}_k) \in \mathbb{R}_{\geq 0}$ is the experienced travel time.
\end{definition}

Note that the mobility inconvenience metric $\phi_i$ increases when $\tilde{\theta}_i(\mathbf{a}_k)$ increases. From a modeling perspective, traveling with time delays or during peak times can cause significant inconveniences to any traveler $i \in \mathcal{I}_k$. Although, an exact form of $\phi$ is beyond the scope of this chapter, our definition of $\phi$ is consistent with general inconvenience functions in the literature \cite{dumas1990,cordeau2007}.

Next, a traveler $i$'s satisfaction is captured by a valuation function $v_i$, which can reflect the traveler's \emph{willingness-to-pay} for their travel, i.e.,
    \begin{equation}\label{EQN:mobility-valuation}
        v_i(\mathbf{a}_k) = \bar{v}_i - \phi_i \left( \alpha_i, \theta_i, \tilde{\theta}_i(\mathbf{a}_k) \right),
    \end{equation}
where $\bar{v}_i \in \mathbb{R}_{\geq 0}$ is the value gained by traveler $i \in \mathcal{I}_k$ when their origin-destination pair $(o_i, d_i)$ is satisfied using service $j \in \mathcal{J}$ without any travel delays, i.e., $\theta_i = \tilde{\theta}_i(\mathbf{a}_k)$. Naturally, for any traveler $i$ and any service $j$, we have $v_i(\mathbf{a}_k) \in [0, \bar{v}_i]$, where $v_i(\mathbf{a}_k) = 0$ means that traveler $i$ is unwilling to use service $j$. Below we summarize the two extreme cases and their interpretation:
    \begin{equation}\label{eqn:cases_valuation}
        v_i(\mathbf{a}_k) =
            \begin{cases}
                \bar{v}_i, & \text{if } \phi_i = 0, \\
                0, & \text{if } \phi_i = \bar{v}_i.
            \end{cases}
    \end{equation}
When $\phi_i = 0$, we say that traveler $i$ travels to their destination in the fastest and most convenient mobility service offered by the mobility market (e.g., a taxicab with no co-travelers). On the other hand, when $\phi_i = \bar{v}_i$, we say that traveler $i$'s personal travel requirements are not satisfied, and the traveler is most inconvenienced with regards to mobility.

Although our analysis can treat the valuation function $v_i$ in its most general form, given by \eqref{EQN:mobility-valuation}, we explicitly define the second component of \eqref{EQN:mobility-valuation} in our mathematical exposition. Thus, the explicit form for the inconvenience mobility metric $\phi_i$ is
    \begin{equation}\label{EQN:explicit-inconvenience-metric}
        \phi_i \left( \alpha_i, \theta_i, \tilde{\theta}_i(\mathbf{a}_k) \right) = \alpha_i \cdot (\tilde{\theta}_i(\mathbf{a}_k) - \theta_i),
    \end{equation}
Basically, \eqref{EQN:explicit-inconvenience-metric} gives the monetary value of the difference between the travel times (experienced vs preferred), and can be interpreted as the travel time tolerance that the traveler can accept (in monetary units).

In our modeling framework, the total utility $u_i(\mathbf{a}_k)$ of traveler $i \in \mathcal{I}_k$, $k = 1, \dots, K$, is given by
    \begin{equation}\label{EQN:traveler-utility}
        u_i(\mathbf{a}_k) = v_i(\mathbf{a}_k) - p_i(\mathbf{a}_k),
    \end{equation}
where $v_i(\mathbf{a}_k)$ is the willingness-to-pay and $p_i(\mathbf{a}_k) \in \mathbb{R}_{\geq 0}$ is the mobility payment that traveler $i$ is required to make to use service $j \in \mathcal{J}$ (e.g., pay road tolls or buy a public transit fare). Hence, \eqref{EQN:traveler-utility} establishes a ``quasi-linear'' relationship between a traveler's satisfaction and payment, both measured in monetary units \cite{mas_colell1995}.

In contrast to the traveler's satisfaction, we also introduce an ``operating cost'' to capture the needed investment that public and private mobility providers and operators make to ensure the proper function of their mobility services.

\begin{definition}\label{DEFN:operating-cost}
    The operating cost of service $j \in \mathcal{J}$ can be computed by
        \begin{equation}\label{eqn:operating_cost}
            c_j(\mathbf{a}_k) = \sum_{i \in \mathcal{I}_k} c_{i j}(a_{i j}),
        \end{equation}
    where $c_{i j}(a_{i j}) \in \mathbb{R}_{\geq 0}$ is traveler $i$'s corresponding share of the operating cost of vehicle $j \in \mathcal{J}$. In the case of $a_{i j} = 0$, we have $c_{i j} = 0$.
\end{definition}

Intuitively, the operating cost $c_{i j}$ captures traveler $i$'s fair share of the costs of service $j \in \mathcal{J}$. These costs can be associated with fuel/energy consumption, drivers' labor reimbursement, maintenance, and environmental impact.

\begin{definition}\label{DEFN:mobility-market}
    Given the traveler-service assignment $\mathbf{a}_k = (a_{i j})_{i \in \mathcal{I}_k, j \in \mathcal{J}}$, the travelers' payments are given by the vector $\mathbf{p}_k = (p_i(a_{i j}))_{i \in \mathcal{I}_k, j \in \mathcal{J}}$. Then, for a subclass $\mathcal{I}_k$, $k = 1, \dots, K$, the proposed mobility market can be fully described by the tuple
        \begin{equation}\label{EQN:mobility-market}
            \big\langle \mathcal{I}_k, \mathcal{J}, (\pi_i)_{i \in \mathcal{I}_k}, (u_i)_{i \in \mathcal{I}_k}, \mathbf{a}_k, \mathbf{p}_k \big\rangle,
        \end{equation}
    where $(\pi_i)_{i \in \mathcal{I}_k}$ is considered private information (unknown to the social planner), and the experienced travel time $\tilde{\theta}_i$ and operation costs $c_j$ of all mobility services are considered known to the social planner.
\end{definition}

Note that in Definition \ref{DEFN:mobility-market}, we have also defined the informational structure of the proposed market. The operation costs $(c_j)_{j \in \mathcal{J}}$ are considered public information as well as the minimum acceptable mobility payments $(\sigma_i)_{i \in \mathcal{I}}$. In general, though, any VCG-based mechanism requires agents to report their entire valuation function \cite{sanghavi2008}. In our case, we can take advantage of more advanced and sophisticated data gathering techniques so that we may infer the form and shape of a traveler's valuation (and utility) function using, for example, historical and empirical data \cite{carrasco2008,abou_zeid2011}. Hence, the functional form of $v_i$ can be considered known, but the realization of $v_i(\cdot)$ is agent $i$'s private information. It is important to note that the evaluation of any traveler $i$'s valuation function can be learned using the three-parameter tuple $\pi_i$, which provides the personal travel requirements of any traveler $i \in \mathcal{I}_k$. In addition, we expect any social planner of a generic transportation system to have the ability (e.g., using regression analysis \cite{ben_akiva2018}) to approximate the experienced travel time of any mobility service and its operating costs. Hence, the only private information that we are required to elicit from the travelers is their personal travel requirements $(\pi_i)_{i \in \mathcal{I}_k}$, $k = 1, \dots, K$. At the same time, receiving communication in the form of messages from all travelers regarding the $(\pi_i)_{i \in \mathcal{I}_k}$, $k = 1, \dots, K$ can be an unrealistic burden. That is why, in our framework, any traveler $i \in \mathcal{I}$ is expected to report the evaluation of their valuation function $v_i$, which depends on their $\pi_i$. Essentially, we parameterize the private information of travelers into a one-dimensional number. In future research, we plan to address a multi-dimensional mechanism to ensure there is no loss of information of the traveler's preferences.

On a different note, a natural question to ask here is whether there is any guarantee that the travelers' mobility payments will meet the providers' operating costs. As we saw in Section \ref{sec:theoretical-preliminaries}, the VCG mechanism can only charge travelers their social cost or impact into the mobility system. Thus, this might lead to very low mobility payments for a significant number of travelers, leading to deficits to cover operating costs for the providers. Since, in reality, we cannot expect any providers to serve travelers indefinitely when their costs have not been met, we introduce a ``pricing base'' for the mobility payments. Essentially, these bases can be chosen by the providers to ensure that no payment by any traveler is below a set value (e.g., minimum acceptable payment), which can be determined approximately by the traveler's location and destination, supply and demand, and operator's reimbursement fee \cite{hall2015}.

\begin{definition}\label{DEFN:sustainability}
    The \emph{minimum acceptable mobility payment} of any service $j \in \mathcal{J}$ is given by $\sigma_i(\mathbf{a}_k) \in \mathbb{R}_{\geq 0}$, for any traveler $i \in \mathcal{I}_k$, $k = 1, \dots, K$. If for an arbitrary traveler $i$, we have $p_i(\mathbf{a}_k) \geq \sigma_i(\mathbf{a}_k)$, then we say that the mobility market, defined in \eqref{EQN:mobility-market}, is \emph{economically sustainable}.
\end{definition}

The minimum acceptable mobility payments $\sigma = (\sigma_i)_{i \in \mathcal{I}}$ are considered public information set by the providers and may be different for each traveler $i \in \mathcal{I}_k$, $k = 1, \dots, K$.

In the modeling framework described above, we impose the following assumption:

\begin{assumption}\label{ASM:strategic-setting}
    For all subclasses $\mathcal{I}_k$, $k = 1, \dots, K$, $K \in \mathbb{N}$, any traveler $i \in \mathcal{I}_k$ is modeled as a selfish decision-maker with private information $\pi_i = (\alpha_i, \theta_i, \bar{v}_i)$. Traveler $i$'s objective is to maximize their total utility $u_i(\mathbf{a}_k) = v_i(\mathbf{a}_k) - p_i(\mathbf{a}_k)$ in a non-cooperative game-theoretic setting.
\end{assumption}

Assumption \ref{ASM:strategic-setting} essentially says that each traveler is selfish in the sense that they are only interested in their own well-being. In economics, such behavior is called ``strategic'' since agents attempt to misreport their private information to the social planner if that means higher individual benefits.

\begin{assumption}\label{ASM:supply_demand}
    The aggregate usage capacities of all mobility services can adequately serve all travel requests of travelers. Mathematically, we have $\sum_{j \in \mathcal{J}} \varepsilon_j = n = |\mathcal{I}|$.
\end{assumption}

Intuitively, Assumption \ref{ASM:supply_demand} ensures that no traveler will remain unassigned. We can justify this assumption as follows: our focus is on efficiently allocating the different mobility services to travelers in a mobility market, a multimodal mobility system that incorporates public transit services with high traveler capacity capabilities. A relaxation of this assumption must consider scenarios where the existing mobility services cannot meet the travelers' demand, thus transforming our problem into a ``mobility and equity'' problem (giving priority to a subset of travelers in a fair way).

\subsection{The Optimization Problem Statement of the Emerging Mobility Market}
\label{subsec:problem-statement}

In the proposed mobility market, travelers request (via a smartphone app), in advance, a travel recommendation from the social planner that satisfies their origin-destination. Given the travelers' origin-destination pairs, the social planner partitions all travelers to different subclasses, as described in Definition \ref{DEFN:traveler-partition}. Thus, travelers from the same neighborhood have the same origin. Similarly, travelers going to the same neighborhood have the same destination. The social planner's task is to elicit the travelers' preferences, attempt to satisfy all travel requests, and provide recommendations to the travelers (e.g., which mobility service to use) by considering the social optimum subject to the city network's physical constraints. Hence, we are interested in minimizing the travel inconvenience of all travelers and the operating costs.

\begin{remark}
    Without loss of generality and to simplify the mathematical analysis in our exposition, we consider that both the mobility inconvenience metrics $(\phi_i)_{i \in \mathcal{I}_k}$, $k = 1, \dots, K$, the minimum mobility payments $(\sigma_i)_{i \in \mathcal{I}_k}$, $k = 1, \dots, K$, and the operating costs $(c_j)_{j \in \mathcal{J}}$ are normalized. This ensures that $\phi_i$, $\sigma_i$, and $c_j$ do not dominate each other in Problem \ref{PROB:centralized-problem} next, while all three are measured in the same monetary units.
\end{remark}

\begin{problem}\label{PROB:centralized-problem}
    For each subclass $\mathcal{I}_k$, $k = 1, \dots, K$, the optimization problem is
        \begin{gather}
            \min_{\mathbf{a}_k} \sum_{i \in \mathcal{I}_k} \left[ \phi_i \left( \alpha_i, \theta_i, \tilde{\theta}_i(\mathbf{a}_k) \right) + \sigma_i(\mathbf{a}_k) \right] + \sum_{j \in \mathcal{J}} c_j(\mathbf{a}_k), \label{EQN:prob1-objective} \\
            \text{subject to:} \notag \\
            \sum_{j \in \mathcal{J}} a_{i j} \leq 1, \quad \forall i \in \mathcal{I}_k, \label{CSTR:user-at-most-one-service} \\
            \sum_{i \in \mathcal{I}_k} a_{i j} \leq \varepsilon_j, \quad \forall j \in \mathcal{J}, \label{CSTR:service-not-violate-capacity} \\
            \sum_{j \in \mathcal{J}_h} \sum_{i \in \mathcal{I}_k} a_{i j} \leq \gamma_e, \quad \forall h \in \{1, 2, \dots, H\}, \quad \forall e \in \mathcal{E}, \label{CSTR:total-car-services-not-violate-capacity}
        \end{gather}
    where \eqref{CSTR:user-at-most-one-service} assures that each traveler $i \in \mathcal{I}_k$ will be assigned at most one mobility service, and \eqref{CSTR:service-not-violate-capacity} stipulates that service $j$'s maximum usage capacity $\varepsilon_j$ must not be exceeded. Lastly, \eqref{CSTR:total-car-services-not-violate-capacity} ensures that there will be no congestion on the links that represent roads or public transit connections. Note also that even though in Problem \ref{PROB:centralized-problem} we focus only on the $k$th partition of the set of travelers $\mathcal{I}$, we do not need to do the same for the mobility services. In other words, since each type of mobility services is associated with a unique link that connects any two nodes, any services that do not satisfy $(o_k, d_k)$ will not be considered in the optimization.
\end{problem}

Problem \ref{PROB:centralized-problem} is similar to the many-to-one assignment problem, and standard algorithmic approaches (e.g., Jonker-Volgenant algorithm \cite{jonker1987}) exist to find its global optimal solution or, in worst-case scenarios, a second-best optimal approximation of a solution. We can also reformulate Problem \ref{PROB:centralized-problem} to a linear program by relaxing to a non-negativity constraint the binary optimization variable $a_{i j}$ for all $i \in \mathcal{I}$ and $j \in \mathcal{J}$. We can then guarantee that an optimal solution of zeros and ones exists by noting that the constraint matrix formed by \eqref{CSTR:user-at-most-one-service} - \eqref{CSTR:total-car-services-not-violate-capacity} satisfies the total unimodularity property \cite{schrijver1998}. Note, though, that these approaches assume complete information of all parameters and variables in the model. Such an assumption is unreasonable to expect from strategic decision-makers, so, in our framework, travelers are not expected to report their private information truthfully. This turns our problems to a \emph{preference elicitation problem}. Our task in Section \ref{sec:methodology} is to provide a theoretical approach that elicits the necessary private information of all travelers using monetary incentives in the form of mobility payments (e.g., tolls, fares, fees).

\section{Methodology for the Design of Mobility Incentives}
\label{sec:methodology}

We can reformulate Problem \ref{PROB:centralized-problem} as a standard social welfare maximization problem. First, recall that $\phi_i \left( \alpha_i, \theta_i, \tilde{\theta}_i(\mathbf{a}_k) \right) = \bar{v}_i - v_i(\mathbf{a}_k)$, so the objective function \eqref{EQN:prob1-objective} becomes
    \begin{equation}
        \max_{\mathbf{a}_k} \sum_{i \in \mathcal{I}_k} \left[ v_i(\mathbf{a}_k) - \sigma_i(\mathbf{a}_k) \right] - \sum_{j \in \mathcal{J}} c_j(\mathbf{a}_k).
    \end{equation}
This reformulation will prove useful as the design of the monetary payments relies on the social welfare impact (or mobility externality) caused by one traveler to the rest of the travelers in the proposed mobility market.

\begin{problem}\label{PROB:equivalency}
    We rewrite Problem \ref{PROB:centralized-problem} as follows:
        \begin{gather}
            \max_{\mathbf{a}_k} \sum_{i \in \mathcal{I}_k} \left[ v_i(\mathbf{a}_k) - \sigma_i(\mathbf{a}_k) \right] - \sum_{j \in \mathcal{J}} c_j(\mathbf{a}_k), \label{EQN:prob2-objective} \\
            \text{subject to:} \notag \\
            \sum_{j \in \mathcal{J}} a_{i j} \leq 1, \quad \forall i \in \mathcal{I}_k, \\
            \sum_{i \in \mathcal{I}_k} a_{i j} \leq \varepsilon_j, \quad \forall j \in \mathcal{J}, \\
            \sum_{j \in \mathcal{J}_h} \sum_{i \in \mathcal{I}_k} a_{i j} \leq \gamma_e, \quad \forall h \in \{1, 2, \dots, H\}, \quad \forall e \in \mathcal{E},
        \end{gather}
    where $\mathbf{a}_k = (a_{i j})_{{i \in \mathcal{I}_k}, j \in \mathcal{J}}$ denotes the solution of Problem \ref{PROB:equivalency}.
\end{problem}

In order for the solution of Problem \ref{PROB:equivalency} to be socially-efficient, we would need a control input in utility function \eqref{EQN:traveler-utility} to incentivize all travelers to report their personal travel requirements truthfully. In our case, this control input is the payments $\mathbf{p}_k$, $k = 1, \dots, K$, which can be designed to be the difference between the \emph{maximum social welfare with traveler $\ell \in \mathcal{I}_k$ not participating} and the \emph{maximum social welfare of other travelers with traveler $\ell$ participating}. Thus, to capture the first term, we revise Problem \ref{PROB:equivalency} by adding constraint \eqref{EQN:prob3-extra} to help us capture the ``mobility externality'' of traveler $\ell$ rejecting any travel recommendations from the social planner. For example, traveler $\ell$ may use a taxicab with no other co-travelers. Thus, Problem \ref{PROB:equivalency} takes the following form.

\begin{problem}\label{PROB:mobility-payments}
    For each traveler $i \in \mathcal{I}_k$, $k = 1, \dots, K$, we fix traveler $\ell \in \mathcal{I}_k$ and solve the following optimization problem:
        \begin{gather}
            \max_{\mathbf{b}_k} \sum_{i \in \mathcal{I}_k} \left[ v_i(\mathbf{b}_k) - \sigma_i(\mathbf{b}_k) \right] - \sum_{j \in \mathcal{J}} c_j(\mathbf{b}_k), \label{EQN:prob3-objective} \\
            \text{subject to:} \notag \\
            \sum_{j \in \mathcal{J}} b_{i j} \leq 1, \quad \forall i \in \mathcal{I}_k, \\
            \sum_{i \in \mathcal{I}_k} b_{i j} \leq \varepsilon_j, \quad \forall j \in \mathcal{J}, \\
            \sum_{j \in \mathcal{J}_h} \sum_{i \in \mathcal{I}_k} b_{i j} \leq \gamma_e, \quad \forall h \in \{1, 2, \dots, H\}, \quad \forall e \in \mathcal{E}, \\
            b_{\ell j} = 0, \quad \forall j \in \mathcal{J}, \label{EQN:prob3-extra}
        \end{gather}
    where $\mathbf{b}_k = (b_{i j})_{{i \in \mathcal{I}_k}, j \in \mathcal{J}}$ defined similarly as in \eqref{EQN:binary-variable} denotes the solution of Problem \ref{PROB:mobility-payments}, and \eqref{EQN:prob3-extra} states that traveler $\ell \in \mathcal{I}_k$ is not considered in the optimization problem.
\end{problem}

\begin{remark}
    In what follows, to simplify the mathematical exposition, we introduce the following notation:
        \begin{align}
            w_2(\mathbf{a}_k) & = \sum_{i \in \mathcal{I}_k} \left[ v_i(\mathbf{a}_k) - \sigma_i(\mathbf{a}_k) \right] - \sum_{j \in \mathcal{J}} c_j(\mathbf{a}_k), \\
            w_3(\mathbf{b}_k) & = \sum_{i \in \mathcal{I}_k} \left[ v_i(\mathbf{b}_k) - \sigma_i(\mathbf{b}_k) \right] - \sum_{j \in \mathcal{J}} c_j(\mathbf{b}_k),
        \end{align}
    where $w_2$ and $w_3$ denote the objective functions of Problems \ref{PROB:equivalency} and \ref{PROB:mobility-payments}, respectively.
\end{remark}

We can now propose the exact form of the mobility payment $p_{\ell}$ for an arbitrary traveler $\ell \in \mathcal{I}_k$, $k = 1, \dots, K$, of the proposed mobility market. For any subclass $\mathcal{I}_k$, $k = 1, \dots, K$, traveler $\ell \in \mathcal{I}_k$ makes the following payment:
    \begin{equation}\label{EQN:pricing-scheme}
        p_{\ell}(\mathbf{a}_k, \mathbf{b}_k) = w_3(\mathbf{b}_k) - \left[ w_2(\mathbf{a}_k) - v_{\ell}(\mathbf{a}_k) \right].
    \end{equation}
Since $w_3(\mathbf{b}_k)$ yields the maximum social welfare from the traveler-service assignment $\mathbf{b}_k$ when traveler $\ell \in \mathcal{I}_k$ does not participate in the mobility market, it can be viewed by traveler $\ell \in \mathcal{I}_k$ in \eqref{EQN:pricing-scheme} as a constant, regardless of what traveler $\ell$ reports to the social planner about their own personal travel requirements $\pi_{\ell}$. The term $\left[ w_2(\mathbf{a}_k) - v_{\ell}(\mathbf{a}_k) \right]$ in \eqref{EQN:pricing-scheme} represents the maximum social welfare of all travelers other than traveler $\ell \in \mathcal{I}_k$, when traveler $\ell \in \mathcal{I}_k$ partakes in the mobility market. As a consequence, $p_{\ell}$ can be interpreted as the externality caused by traveler $\ell \in \mathcal{I}_k$ to all other travelers. In addition, the computation of the mobility payments \eqref{EQN:pricing-scheme} requires solving Problem \ref{PROB:mobility-payments} repeatedly for each traveler. As shown in Algorithm \ref{ALG:pricing-mechanism}, first we derive the optimal solution of Problem \ref{PROB:equivalency}, and then we use the optimal solution of Problem \ref{PROB:mobility-payments} to compute the monetary payment of each traveler $\ell \in \mathcal{I}_k$.

\begin{algorithm}[ht]
\label{ALG:pricing-mechanism}
    \SetAlgoLined
    \KwData{$\mathcal{I}_k, \mathcal{J}, (\pi_i)_{i \in \mathcal{I}_k}, (u_i)_{i \in \mathcal{I}_k}$}
    \KwResult{$\mathbf{a}_k ^ *$ and $\mathbf{p}_k$}
        Solve for the optimal solution $\mathbf{a}_k ^ *$ of Problem \ref{PROB:equivalency}\;
    \For{$\ell \in \mathcal{I}_k$}{
    Solve for the optimal solution $\mathbf{b}_k ^ *$ of Problem \ref{PROB:mobility-payments}\;
    Next, compute
        \begin{equation*}
            p_{\ell}(\mathbf{a}_k ^ *, \mathbf{b}_k ^ *) = w_3(\mathbf{b}_k ^ *) - \left[ w_2(\mathbf{a}_k ^ *) - v_{\ell}(\mathbf{a}_k ^ *) \right].
        \end{equation*}
    }
    \caption{Solution of Problem \ref{PROB:equivalency} with Problems \ref{PROB:mobility-payments}}
\end{algorithm}

Before we move on to the next section, we note that informally we talked about a traveler not participating in the mobility market in solving Problem \ref{PROB:mobility-payments}. This idea helps us design the mobility payments in \eqref{EQN:pricing-scheme} by identifying the mobility externalities in the welfare of all travelers. Thus, we introduce the notion of ``mobility exclusion,'' which will help us capture the socioeconomic impact of any traveler on the rest of the mobility market.

\begin{definition}\label{DEFN:exclusion}
    For any subclass $\mathcal{I}_k$, $k = 1, \dots, K$, given a traveler-service assignment $\mathbf{a}_k$ of Problem \ref{PROB:equivalency}, a traveler $\ell \in \mathcal{I}_k$ is said to be \emph{mobility excluded} if they are not assigned to any mobility service in the traveler-service assignment $\mathbf{b}_k$ of Problem \ref{PROB:mobility-payments}.
\end{definition}

Problem \ref{PROB:mobility-payments} is used to compute the mobility payments for each traveler in the mobility market by identifying the mobility externality caused by the decision-making of the traveler to the rest of the market. In addition, however, we are also interested in identifying the traveler's impact on (i) operating costs and (ii) overall welfare. We shall see in the next section how we can achieve this.

\section{Properties of the Mobility Market}
\label{sec:properties}

Our first result is an immediate and straightforward consequence of Definition \ref{DEFN:exclusion}. Recall that the operating cost $c_{i j}(a_{i j})$ captures traveler $i$'s fair share of the mobility service $j$'s costs that they use under the traveler-service assignment $\mathbf{a}_k$.

\begin{corollary}\label{COR:transportation-costs-relation}
    Let $\mathbf{b}_k ^ {\ell}$ be a feasible traveler-service assignment of Problem \ref{PROB:mobility-payments}. Given that traveler $\ell \in \mathcal{I}_k$ is mobility excluded, the operating cost that is associated with the traveler-service assignment $\mathbf{b}_k ^ {\ell}$ is smaller than or equal than the operating cost associated with the optimal assignment $\mathbf{a}_k ^ *$ of Problem \ref{PROB:equivalency}, i.e., we have
        \begin{equation}
            \sum_{i \in \mathcal{I}_k} c_{i j}(a_{i j} ^ *) \geq \sum_{i \in \mathcal{I}_k \setminus \{\ell\}} c_{i j}(b_{i j} ^ {\ell}).
        \end{equation}
\end{corollary}

Similarly, using Definition \ref{DEFN:exclusion}, we show that the sum of valuations (or welfare) of all travelers other than the traveler, who is mobility excluded specifically in Problem \ref{PROB:mobility-payments}, is greater or equal than the sum of valuations evaluated at the traveler-service assignment of Problem \ref{PROB:equivalency}.

\begin{lemma}\label{LEM:valuation-relation-priority}
    Let $\mathbf{b}_k ^ {\ell}$ be a feasible traveler-service assignment of Problem \ref{PROB:mobility-payments}, in which traveler $\ell \in \mathcal{I}_k$ is mobility excluded. Then, we have
        \begin{equation}\label{EQN:given-priority-relation}
            \sum_{i \in \mathcal{I}_k \setminus \{\ell\}} v_i(\mathbf{a}_k) \leq \sum_{i \in \mathcal{I}_k} v_i(\mathbf{b}_k ^ {\ell}).
        \end{equation}
\end{lemma}

\begin{proof}
    Given that traveler $\ell \in \mathcal{I}_k$ is mobility excluded in the traveler-service assignment $\mathbf{b}_k ^ {\ell}$ of Problem \ref{PROB:mobility-payments}, we know that there is one less traveler required to be served by any mobility service in the market. Naturally, this affects the experienced travel times of any other traveler $i \in \mathcal{I}_k$, i.e., we have either a decreased or constant $\tilde{\theta}_i(\mathbf{b}_k ^ {\ell})$. So, mathematically this means that with traveler-service assignment $\mathbf{a}_k$ of Problem \ref{PROB:equivalency}, we have
        \begin{equation}\label{EQN:lemETT-first-relation}
            \tilde{\theta}_i(\mathbf{b}_k ^ {\ell}) \leq \tilde{\theta}_i(\mathbf{a}_k),
        \end{equation}
    where $\tilde{\theta}_i(\mathbf{b}_k ^ {\ell})$ is the experienced travel time of traveler $i \in \mathcal{I}_k$ evaluated at $\mathbf{b}_k ^ {\ell}$ and $\tilde{\theta}_i(\mathbf{a}_k)$ is the experienced travel time of traveler $i$ evaluated at $\mathbf{a}_k$. Intuitively, \eqref{EQN:lemETT-first-relation} means there is one less traveler leading to better travel times for other travelers (better here means less). Hence, since the explicit form of traveler $i$'s valuation is given by
        \begin{equation}
            v_i(\mathbf{a}_k) = \bar{v}_i - \phi_i \left( \alpha_i, \theta_i, \tilde{\theta}_i(\mathbf{a}_k) \right)
            = \bar{v}_i - \alpha_i \cdot (\tilde{\theta}_i(\mathbf{a}_k) - \theta_i),
        \end{equation}
    if we compare the two valuations $v_i(\mathbf{a}_k)$ and $v_i(\mathbf{b}_k ^ {\ell})$, we get $v_i(\mathbf{a}_k) \leq v_i(\mathbf{b}_k ^ {\ell})$. This completes the proof.
\end{proof}

Next, we show that for any traveler, their valuation will always be greater or equal than the minimum mobility payment. This will be instrumental in our attempt to show individual rationality later on.

\begin{lemma}\label{LEM:economic-sustainability}
    Let $\mathbf{a}_k ^ *$ denote the optimal solution of Problem \ref{PROB:equivalency}. Then the minimum mobility payment $\sigma_{\ell}$ in the objective function \eqref{EQN:prob2-objective} of Problem \ref{PROB:equivalency} ensures that, for any $\ell \in \mathcal{I}_k$, $k = 1, \dots, K$, $v_{\ell}(\mathbf{a}_k ^ *) \geq \sigma_{\ell}(\mathbf{a}_k ^ *)$.
\end{lemma}

\begin{proof}
    Let $\mathbf{a}_k ^ *$ denote the optimal solution of Problem \ref{PROB:equivalency} and ${\mathbf{b}_k ^ {\ell}} ^ *$ the corresponding solution of Problem \ref{PROB:mobility-payments}. Hence, traveler $\ell$ has been assigned a mobility service in the optimal traveler-service assignment $\mathbf{a}_k ^ *$, but they are mobility excluded in ${\mathbf{b}_k ^ {\ell}} ^ *$. Thus, we have
        \begin{align}
            w_3({\mathbf{b}_k ^ {\ell}} ^ *) & = \sum_{i \in \mathcal{I}_k} \left[ v_i({\mathbf{b}_k ^ {\ell}} ^ *) - \sigma_i({\mathbf{b}_k ^ {\ell}} ^ *) \right] - \sum_{j \in \mathcal{J}} c_j({\mathbf{b}_k ^ {\ell}} ^ *) \notag \\
            & \geq \sum_{i \in \mathcal{I}_k \setminus \{\ell\}} v_i(\mathbf{a}_k ^ *) - \sum_{i \in \mathcal{I}_k} \sigma_i({\mathbf{b}_k ^ {\ell}} ^ *) - \sum_{j \in \mathcal{J}} c_j(\mathbf{a}_k ^ *), \label{EQN:lemPC-first}
        \end{align}
    where \eqref{EQN:lemPC-first} follows from Corollary \ref{COR:transportation-costs-relation} and Lemma \ref{LEM:valuation-relation-priority}. Next, we look at the welfare of an arbitrary traveler $i \in \mathcal{I}_k$ under $\mathbf{a}_k ^ *$, i.e.,
        \begin{align}
            w_2(\mathbf{a}_k ^ *) & = \sum_{i \in \mathcal{I}_k} \left[ v_i(\mathbf{a}_k ^ *) - \sigma_i(\mathbf{a}_k ^ *) \right] - \sum_{j \in \mathcal{J}} c_j(\mathbf{a}_k ^ *) \notag \\
            & = \sum_{i \in \mathcal{I}_k} v_i(\mathbf{a}_k ^ *) - \sum_{i \in \mathcal{I}_k} \sigma_i(\mathbf{a}_k ^ *) - \sum_{j \in \mathcal{J}} c_j(\mathbf{a}_k ^ *), \label{EQN:lemPC-second}
        \end{align}
    where it also follows that $w_2(\mathbf{a}_k ^ *) \geq w_3({\mathbf{b}_k ^ {\ell}} ^ *)$ from the fact that ${\mathbf{b}_k ^ {\ell}} ^ *$ is not an optimal solution of Problem \ref{PROB:equivalency}. Thus, if we compare \eqref{EQN:lemPC-first} and \eqref{EQN:lemPC-second}, we get
        \begin{equation}
            \sum_{i \in \mathcal{I}_k} v_i(\mathbf{a}_k ^ *) - \sum_{i \in \mathcal{I}_k} \sigma_i(\mathbf{a}_k ^ *) - \sum_{j \in \mathcal{J}} c_j(\mathbf{a}_k ^ *) \geq \sum_{i \in \mathcal{I}_k \setminus \{\ell\}} v_i(\mathbf{a}_k ^ *) - \sum_{i \in \mathcal{I}_k} \sigma_i({\mathbf{b}_k ^ {\ell}} ^ *) - \sum_{j \in \mathcal{J}} c_j(\mathbf{a}_k ^ *). \label{EQN:lemPC-one-before-last}
        \end{equation}
    So, by simplifying and rearranging \eqref{EQN:lemPC-one-before-last}, we have
        \begin{align}
            \sum_{i \in \mathcal{I}_k} v_i(\mathbf{a}_k ^ *) - \sum_{i \in \mathcal{I}_k \setminus \{\ell\}} v_i(\mathbf{a}_k ^ *) & \geq \sum_{i \in \mathcal{I}_k} \sigma_i(\mathbf{a}_k ^ *) - \sum_{i \in \mathcal{I}_k} \sigma_i({\mathbf{b}_k ^ {\ell}} ^ *), \notag \\
            & = \sigma_{\ell}(\mathbf{a}_k ^ *) - \sigma_{\ell}({\mathbf{b}_k ^ {\ell}} ^ *) = \sigma_{\ell}(\mathbf{a}_k ^ *), \label{EQN:lemPC-last}
        \end{align}
    since $\sigma_{\ell}({\mathbf{b}_k ^ {\ell}} ^ *) = 0$ as traveler $\ell$ is not assigned any mobility service under the traveler-service assignment ${\mathbf{b}_k ^ {\ell}} ^ *$. Therefore, \eqref{EQN:lemPC-last} simplifies to $v_{\ell}(\mathbf{a}_k ^ *) \geq \sigma_{\ell}(\mathbf{a}_k ^ *)$.
\end{proof}

Our first main result is incentive compatibility, which means that all travelers are incentivized to report their private information truthfully. Formally, for an arbitrary traveler $i \in \mathcal{I}_k$, $k = 1, \dots, K$, given that $u_i '$ is the utility gained with misreported $\pi_i$ and $u_i$ is the ``actual'' utility, showing that $u_i ' \leq u_i$ guarantees truthfulness.

\begin{theorem}\label{THM:IC}
    The mobility market defined in \eqref{EQN:mobility-market} provides the appropriate monetary incentives to each traveler $i \in \mathcal{I}_k$, $k = 1, \dots, K$ to report their personal travel requirements $\pi_i = (\alpha_i, \theta_i, \bar{v}_i)$ truthfully regardless of what other travelers report.
\end{theorem}

\begin{proof}
    It is sufficient to show incentive compatibility only for an arbitrary mobility market for some arbitrary $k \in \{1, \dots, K\}$. Suppose some traveler $\ell \in \mathcal{I}_k$ misreports their personal travel requirements denoted by $\pi_{\ell} = (\alpha_{\ell} ', \theta_{\ell} ', \bar{v}_{\ell} ')$ to the social planner. Thus, we have
        \begin{equation}
            v_{\ell} ' (\mathbf{a}_k) = \bar{v}_{\ell} ' - \phi_\ell \left( \alpha_\ell ', \theta_{\ell} ', \tilde{\theta}_\ell(\mathbf{a}_k) \right).
        \end{equation}
    The objective function of Problem \ref{PROB:equivalency} becomes
        \begin{equation}\label{EQN:one-traveler-misreports}
            w_2 ' (\mathbf{a}_k) = \sum_{i \in \mathcal{I}_k \setminus \{\ell\}} \left[ v_i(\mathbf{a}_k) - \sigma_i(\mathbf{a}_k) \right] - \sum_{j \in \mathcal{J}} c_j(\mathbf{a}_k) + v_{\ell} ' (\mathbf{a}_k),
        \end{equation}
    where the feasible solution of \eqref{EQN:one-traveler-misreports} is subject to the same constraints as in Problem \ref{PROB:equivalency}. 
    %
    We denote the optimal solution of the optimization problem that traveler $\ell$ has misreported their personal travel requirements $\pi_{\ell}$ with \eqref{EQN:one-traveler-misreports} as the objective function by $\tilde{\mathbf{a}}_k ^ *$. Then, for traveler $\ell \in \mathcal{I}_k$ their mobility payment can be computed as follows:
        \begin{equation}\label{EQN:payment-traveler-misreports}
            p_{\ell} ' (\tilde{\mathbf{a}}_k ^ *, \tilde{\mathbf{b}}_k ^ *) = w_3(\tilde{\mathbf{b}}_k ^ *) - \left[ w_2 ^ {\ell}(\tilde{\mathbf{a}}_k ^ *) - v_{\ell} ' (\tilde{\mathbf{a}}_k ^ *) \right] = w_3(\mathbf{b}_k ^ *) - \left[ w_2 ^ {\ell}(\tilde{\mathbf{a}}_k ^ *) - v_{\ell} ' (\tilde{\mathbf{a}}_k ^ *) \right],
        \end{equation}
    where $\tilde{\mathbf{b}}_k ^ *$ denotes the optimal solution of Problem \ref{PROB:mobility-payments} with traveler $\ell \in \mathcal{I}_k$ misreporting. However, $w_3(\tilde{\mathbf{b}}_k ^ *) = w_3(\mathbf{b}_k ^ *)$ since, in Problem \ref{PROB:mobility-payments}, it does not matter what traveler $\ell \in \mathcal{I}_k$ reports. Thus, the total utility of traveler $\ell \in \mathcal{I}_k$ is
        \begin{equation}\label{EQN:utility-traveler-misreports}
            u_{\ell}(\tilde{\mathbf{a}}_k ^ *) = v_{\ell}(\tilde{\mathbf{a}}_k ^ *) - p_{\ell} ' (\tilde{\mathbf{a}}_k ^ *, \mathbf{b}_k ^ *),
        \end{equation}
    where for traveler $\ell \in \mathcal{I}_k$ the term $v_{\ell}(\tilde{\mathbf{a}}_k ^ *)$ is the actual satisfaction gained by misreporting their private information. Substituting \eqref{EQN:payment-traveler-misreports} into \eqref{EQN:utility-traveler-misreports} yields
        \begin{equation}
            u_{\ell}(\tilde{\mathbf{a}}_k ^ *) = v_{\ell} (\tilde{\mathbf{a}}_k ^ *) - \left[ w_3(\mathbf{b}_k ^ *) - \left( w_2 ^ {\ell}(\tilde{\mathbf{a}}_k ^ *) - v_{\ell} ' (\tilde{\mathbf{a}}_k ^ *) \right) \right],
        \end{equation}
    which after a few simplifications gives
        \begin{equation}\label{EQN:thmIC-one-before-last}
            u_{\ell}(\tilde{\mathbf{a}}_k ^ *) = v_{\ell} (\tilde{\mathbf{a}}_k ^ *) - w_3(\mathbf{b}_k ^ *) - \left[ \left( \sum_{i \in \mathcal{I}_k \setminus \{\ell\}} \left[ v_i(\tilde{\mathbf{a}}_k ^ *) - \sigma_i(\tilde{\mathbf{a}}_k ^ *) \right] - \sum_{j \in \mathcal{J}} c_j(\tilde{\mathbf{a}}_k ^ *) + v_{\ell} ' (\tilde{\mathbf{a}}_k ^ *) \right) - v_{\ell} ' (\tilde{\mathbf{a}}_k ^ *) \right].
        \end{equation}
    Hence, as the term $v_{\ell} ' (\tilde{\mathbf{a}}_k ^ *)$ appears in opposite signs in \eqref{EQN:thmIC-one-before-last}, we have 
        \begin{align}
            u_{\ell}(\tilde{\mathbf{a}}_k ^ *) & = \left[ \sum_{i \in \mathcal{I}_k} \left[ v_i(\tilde{\mathbf{a}}_k ^ *) - \sigma_i(\tilde{\mathbf{a}}_k ^ *) \right] - \sum_{j \in \mathcal{J}} c_j(\tilde{\mathbf{a}}_k ^ *) \right] - w_3(\mathbf{b}_k ^ *) \notag \\
            & = w_2(\tilde{\mathbf{a}}_k ^ *) - w_3(\mathbf{b}_k ^ *).
        \end{align}
    Note that $\tilde{\mathbf{a}}_k ^ *$ is not necessarily the optimal solution of Problem \ref{PROB:equivalency}. Thus, we have $w_2(\tilde{\mathbf{a}}_k ^ *) \leq w_2(\mathbf{a}_k ^ *)$. So, we observe that
        \begin{align}\label{EQN:thmIC-last-argument}
            u_{\ell}(\tilde{\mathbf{a}}_k ^ *) = w_2(\tilde{\mathbf{a}}_k ^ *) - w_3(\mathbf{b}_k ^ *) \leq w_2(\mathbf{a}_k ^ *) - w_3(\mathbf{b}_k ^ *) = u_{\ell}(\mathbf{a}_k ^ *).
        \end{align}
    Therefore, from \eqref{EQN:thmIC-last-argument}, it follows immediately that the proposed mobility market is incentive compatible.
\end{proof}

Our next result is individual rationality, which implies that all travelers voluntarily participate in the proposed mobility market. Formally, for any traveler $i \in \mathcal{I}_k$, $k = 1, \dots, K$, if traveler $i$'s utility $u_i$ is non-negative, i.e., $u_i \geq 0$, then we say traveler $i$ voluntarily participates in the mobility market. This is important as we can guarantee for any traveler $i$ that what they are willing to pay, $v_i$, will never be less than what they actually pay, $p_i$.

\begin{theorem}
    The mobility market is individually rational. For any subclass $\mathcal{I}_k$, $k = 1, \dots, K$, and for any traveler $i \in \mathcal{I}_k$, the utility of any traveler is non-negative, i.e., we have for all $i \in \mathcal{I}_k$, $u_i(\mathbf{a}_k) \geq 0$. Equivalently, $v_i(\mathbf{a}_k) \geq p_i(\mathbf{a}_k)$.
\end{theorem}

\begin{proof}
    It is sufficient to show the result only for one instance of a mobility market for some $k = \{1, \dots, K\}$. There are two cases to consider. First, let us suppose that traveler $\ell \in \mathcal{I}_k$ rejects any travel recommendations from the social planner; denote such an assignment by $\hat{\mathbf{a}}_k$. From \eqref{EQN:pricing-scheme}, traveler $\ell$ would be required to make a monetary payment equal to their maximum willingness-to-pay, i.e., $p_{\ell} = \bar{v}_{\ell}$. This implies that $u_{\ell}(\hat{\mathbf{a}}_k) = 0$. This is justifiable as traveler ${\ell}$ seeks to travel and the only alternative travel option to our mobility market is a taxicab service.
    
    For the second case, let us consider the utility of an arbitrary traveler $i \in \mathcal{I}_k$ evaluated at the optimal solution $\mathbf{a}_k ^ *$ is given by
        \begin{equation}\label{EQN:thmIR-first}
            u_i(\mathbf{a}_k ^ *) = v_i(\mathbf{a}_k ^ *) - p_i(\mathbf{a}_k ^ *, \mathbf{b}_k ^ *).
        \end{equation}
    Note that by Theorem \ref{THM:IC} all travelers report their true private information at equilibrium. So, substituting \eqref{EQN:pricing-scheme} into \eqref{EQN:thmIR-first} yields
        \begin{equation}
            u_i(\mathbf{a}_k ^ *) = v_i(\mathbf{a}_k ^ *) - \left[ w_3(\mathbf{b}_k ^ *) - \left[ w_2(\mathbf{a}_k ^ *) - v_i(\mathbf{a}_k ^ *) \right] \right] = w_2(\mathbf{a}_k ^ *) - w_3(\mathbf{b}_k ^ *).
        \end{equation}
    Note that for each $k = 1, \dots, K$, the feasible regions of Problems \ref{PROB:equivalency} and \ref{PROB:mobility-payments}, say $\mathcal{F}_2$ and $\mathcal{F}_3$, respectively, satisfy the relation $\mathcal{F}_3 \subset \mathcal{F}_2$. This is because Problem \ref{PROB:mobility-payments} has the exact same constraints plus an additional one, i.e., \eqref{EQN:prob3-extra}, thus the maximization of $w_3$ (which is almost similar to the one in Problem \ref{PROB:equivalency}) will always be less or equal than the maximization of $w_2$. Hence, it follows that $u_i(\mathbf{a}_k ^ *) = w_2(\mathbf{a}_k ^ *) - w_3(\mathbf{b}_k ^ *) \geq 0$. Therefore, the result follows.
\end{proof}

Next, we establish that the proposed mobility market is economically sustainable (see Definition \ref{DEFN:sustainability}).

\begin{theorem}
    The mobility market is economically sustainable, i.e., it is guaranteed to generate revenue from each traveler and always meet the minimum acceptable mobility payments. In other words, for each subclass $\mathcal{I}_k$, $k = 1, \dots, K$, and for an arbitrary $\ell \in \mathcal{I}_k$, we have
        \begin{equation}
            p_{\ell}(\mathbf{a}_k ^ *, \mathbf{b}_k ^ *) = w_3(\mathbf{b}_k ^ *) - \left[ w_2(\mathbf{a}_k ^ *) - v_{\ell}(\mathbf{a}_k ^ *) \right] \geq \sigma_{\ell}(\mathbf{a}_k ^ *).
        \end{equation}
\end{theorem}

\begin{proof}
    Let $\mathbf{b}_k ^ *$ be an optimal solution of Problem \ref{PROB:mobility-payments} and ${\mathbf{b}_k ^ {\ell}} ^ *$ be the corresponding feasible solution of Problem \ref{PROB:mobility-payments} with $\mathbf{a}_k ^ *$ an optimal solution of Problem \ref{PROB:centralized-problem}. Since ${\mathbf{b}_k ^ {\ell}} ^ *$ is only a feasible solution, we have
        \begin{equation}\label{EQN:thmREV-first}
            w_3(\mathbf{b}_k ^ *) \geq w_3({\mathbf{b}_k ^ {\ell}} ^ *).
        \end{equation}
    Given the mobility payments \eqref{EQN:pricing-scheme}, if we subtract the term $\left[ w_2(\mathbf{a}_k ^ *) - v_{\ell}(\mathbf{a}_k ^ *) \right]$ from both sides of \eqref{EQN:thmREV-first}, we have
        \begin{equation}\label{EQN:thmREV-payment-comparison}
            p_{\ell}(\mathbf{a}_k ^ *, \mathbf{b}_k ^ *) = w_3(\mathbf{b}_k ^ *) - \left[ w_2(\mathbf{a}_k ^ *) - v_{\ell}(\mathbf{a}_k ^ *) \right] \geq w_3({\mathbf{b}_k ^ {\ell}} ^ *) - \left[ w_2(\mathbf{a}_k ^ *) - v_{\ell}(\mathbf{a}_k ^ *) \right].
        \end{equation}
    The RHS of \eqref{EQN:thmREV-payment-comparison} can be expanded as follows:
        \begin{multline}\label{EQN:thmREV-payment-comparison-2}
            w_3({\mathbf{b}_k ^ {\ell}} ^ *) - \left[ w_2(\mathbf{a}_k ^ *) - v_{\ell}(\mathbf{a}_k ^ *) \right] = \sum_{i \in \mathcal{I}_k} \left[ v_i({\mathbf{b}_k ^ {\ell}} ^ *) - \sigma_i({\mathbf{b}_k ^ {\ell}} ^ *) \right] - \sum_{j \in \mathcal{J}} c_j({\mathbf{b}_k ^ {\ell}} ^ *) \\
            - \left[ \left( \sum_{i \in \mathcal{I}_k} \left[ v_i(\mathbf{a}_k ^ *) - \sigma_i(\mathbf{a}_k ^ *) \right] - \sum_{j \in \mathcal{J}} c_j(\mathbf{a}_k ^ *) \right) - v_{\ell}(\mathbf{a}_k ^ *) \right].
        \end{multline}
%
    After a few simplifications and rearranging of \eqref{EQN:thmREV-payment-comparison-2}, we have
        \begin{equation}\label{EQN:thmREV-relation5}
            p_{\ell}(\mathbf{a}_k ^ *, \mathbf{b}_k ^ *) \geq \left[ \sum_{i \in \mathcal{I}_k} v_i({\mathbf{b}_k ^ {\ell}} ^ *) - \sum_{i \in \mathcal{I}_k \setminus \{\ell\}} v_i(\mathbf{a}_k ^ *) \right] + \sum_{i \in \mathcal{I}_k} \left[ \sigma_i(\mathbf{a}_k ^ *) - \sigma_i({\mathbf{b}_k ^ {\ell}} ^ *) \right] + \left[ \sum_{j \in \mathcal{J}} c_j(\mathbf{a}_k ^ *) - \sum_{j \in \mathcal{J}} c_j({\mathbf{b}_k ^ {\ell}} ^ *) \right].
        \end{equation}
    So, by Corollary \ref{COR:transportation-costs-relation}, the last term in \eqref{EQN:thmREV-relation5} is non-negative. Similarly, by Lemma \ref{LEM:valuation-relation-priority}, the first term in \eqref{EQN:thmREV-relation5} is non-negative. So, we get
        \begin{equation}
            p_{\ell}(\mathbf{a}_k ^ *, \mathbf{b}_k ^ *) \geq \sigma_{\ell}(\mathbf{a}_k ^ *) - \sigma_{\ell}({\mathbf{b}_k ^ {\ell}} ^ *) = \sigma_{\ell}(\mathbf{a}_k ^ *),
        \end{equation}
    since under ${\mathbf{b}_k ^ {\ell}} ^ *$ traveler $\ell$ has not been assigned any mobility service, thus $\sigma_{\ell}({\mathbf{b}_k ^ {\ell}} ^ *) = 0$, and so the result follows immediately.
\end{proof}

\section{Conclusion}
\label{sec:conclusion}

This chapter demonstrates how we can model and study the mobility decision-making of selfish travelers who are faced with the dilemma of ``which mode of transportation to use'' as an economically-inspired mobility market. First, the proposed market provides a \emph{socially-efficient solution}, i.e., the endmost collective travel recommendation respects and satisfies the travelers' preferences regarding mobility and ensures that, implicitly, there will be an alleviation of congestion in the system. We achieve the latter by introducing appropriate constraints in the optimization problem; thus, our solution efficiently allocates all the available mobility services to the travelers. Furthermore, we showed that the proposed mobility market attains the properties of \emph{incentive compatibility} and \emph{individual rationality}. In other words, all travelers are incentivized to participate in the market while also truthfully reporting their personal travel requirements. Last, we introduced the notion of \emph{minimum acceptable mobility payments} to ensure that the tolls and fares collected by the social planner will meet the mobility services' operating costs. Hence, the proposed market satisfies a status of \emph{economic sustainability}.

One particular limitation of the proposed mobility market is that we require all travelers to book in advance, so the traveler-service assignment is static. This implies that the social planner would have to recompute all optimization problems in the mobility market to get an updated traveler-service assignment if the travelers' information changes. However, the static aspect of the proposed model is quite fitting in our case as our aim was to design a mobility market that considers the travelers' personal travel requirements to provide a socially-efficient assignment, i.e., ``who should use which mode of transportation.'' Future work will focus on translating our model and results in a real-time environment. Furthermore, we have implicitly assumed that the travelers' utilities are not interdependent, i.e., a traveler's utility does not depend on the other travelers' private information. It remains an open problem the design of dynamic mechanisms with interdependent utility functions for mobility systems.

Ongoing work includes extending and enhancing the traveler-behavioral model, motivated by a social-mobility survey. Our objective is to observe any correlations between behavioral tendencies or attitudes of travelers and their mode of transportation preference (including CAVs). For example, how likely are people to use CAVs instead of public transit? Will CAVs impact travelers' tendencies and behavior; if yes, then in what way? Answers can help us refine the proposed mobility market and improve our understanding of the socioeconomic impact of CAVs. Our future research efforts will also focus on using methods, techniques, and insights from behavioral economics and mixed integer optimization theory to develop a holistic framework of the societal impact of connectivity and automation in mobility and provide socially-efficient, real-time solutions while tackling any potential rebound effects.

\end{document}